\documentclass[aps,pra,amsmath,notitlepage,nofootinbib,twocolumn,superscriptaddress,noeprint]{revtex4-2}
\usepackage{amsmath,amsfonts,amssymb,amsthm}
\usepackage{mathtools}
\usepackage[pdftex]{graphicx}
\usepackage[usenames, dvipsnames, svgnames, table]{xcolor}
\usepackage[colorlinks=true, citecolor=Blue, linkcolor=BrickRed, urlcolor=Brown]{hyperref}
\usepackage{txfonts}
\usepackage{mathrsfs}
\usepackage{bm}
\usepackage{multirow}
\usepackage{braket}
\usepackage{import}
\usepackage[ruled, noline, noend]{algorithm2e}
\usepackage{array}
\usepackage{comment}
\usepackage{tikz}
\usetikzlibrary{quantikz}
\usepackage{graphicx}
\usepackage[normalem]{ulem}
\useunder{\uline}{\ul}{}

\newcommand{\AJ}[1]{\textcolor{black}{#1}}

\newtheorem{theorem}{Theorem}
\newtheorem{lemma}{Lemma}
\newtheorem{Note}{Note}

\theoremstyle{definition}
\newtheorem{definition}{Definition}
\newcommand{\be}{\begin{equation}}
\newcommand{\ee}{\end{equation}}
\newcommand{\ben}{\begin{eqnarray}}
\newcommand{\een}{\end{eqnarray}}
\newcommand{\bes}{\begin{subequations}}
\newcommand{\ees}{\end{subequations}}
\newcommand{\bF}{\begin{figure}}
\newcommand{\eF}{\end{figure}}

\newcommand{\orcid}[1]{\href{https://orcid.org/#1}{\includegraphics[height = 2ex]{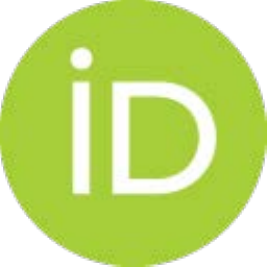}}}

\begin{document}

\title{Accreditation Against Limited Adversarial Noise}

\author{Andrew Jackson \orcid{0000-0002-5981-1604}}
\email{Andrew.J.Jackson@ed.ac.uk}
\affiliation{Department of Physics, University of Warwick, Coventry CV4 7AL, United Kingdom}
\affiliation{School of Informatics, University of Edinburgh, Edinburgh, EH8 9AB, United Kingdom}


\begin{abstract}
I present an accreditation protocol (a variety of quantum verification) where error is assumed to be adversarial (in contrast to the assumption error is implemented by identical CPTP maps used in previous accreditation protocols) -- albeit slightly modified to reflect physically motivated error assumptions. This is achieved by upgrading a pre-existing accreditation protocol (from [S. Ferracin \emph{et al}. Phys. Rev. A 104, 042603 (2021)]) to function correctly in the face of adversarial error, with no diminution in efficiency or suitability for near-term usage. 
\end{abstract}

\date{\today}
\maketitle

 \section{Introduction}
Without methods to assess the quality of quantum computers and computations, they will be untrustworthy and thus almost completely useless.
There are a range of methods for establishing trust in quantum computers and computations, each approaching the problem from a different perspective. Together, these are known as quantum characterization, verification, and validation~\cite{blumekohout2025quantumcharacterizationverificationvalidation} (often abbreviated QCVV). QCVV can be decomposed broadly into:
\begin{enumerate}
    \item Benchmarking~\cite{Proctor2024iph, Hashim2024dox}, which investigates and quantifies the performance of the quantum computer itself. Particularly popular is randomized benchmarking~\cite{PhysRevA.77.012307, Dankert_2009, PhysRevLett.123.060501, PRXQuantum.3.020357}.
    \item The characterization of quantum devices~\cite{blumekohout2025quantumcharacterizationverificationvalidation}, which seeks descriptions of the behavior of quantum devices and their capabilities.
    \item Verification~\cite{Gheorghiu_2018, Broadbent_2018, 8555111, kapourniotis2022unifying}, which investigates and quantifies how well \emph{specific} computations\footnote{E.g. preparing a six-qubit Greenberger-Horne-Zeilinger (GHZ) state~\cite{Greenberger1989}.} are performed.
\end{enumerate}
All of these approaches are vital to the future development and usage of quantum computers, particularly in the NISQ era~\cite{Preskill_2018} -- when quantum computations will be unreliable due to interactions with the environment~\cite{nielsen_chuang_2010} that induce erroneous operators in a computation -- but will remain so into the fault-tolerant era, to quantify the quality of the outputs and monitor the performance of the error correction.

\AJ{However, this paper -- from here on -- focuses exclusively on verification, particularly near-term applicable verification methods that quantify the effect of error rather than just detecting if error occurs. Among such verification protocols, two stand out: accreditation and mirror circuit fidelity estimation~\cite{Proctor2022cep}.}

\AJ{Typically, verification protocols run a quantum computation, allowing for the possibility of error, and -- via some machinations -- decide whether to accept or reject based on if error is detected \emph{at all}. However, detecting only \emph{if} error occurs is a much more limited measure of how well a specific computation is performed than either the ideal-actual variation distance (as in Def.~\ref{def:idealActualVarDist}), which accreditation protocols provide, or fidelity, which mirror circuit fidelity estimation provides.
As such, these protocols will be more useful in the near-term, when \emph{some} error will almost always occur (in large computations), and will allow for more nuanced decisions on how useful results are based on how severe the error is.} 

\begin{definition}
    \label{def:idealActualVarDist}
    For any circuit, $\mathcal{C}$, and an execution of it, $\Tilde{\mathcal{C}}$, the \underline{ideal-actual variation distance} of that execution (denoted as $\nu \big[ \Tilde{\mathcal{C}} \big]$) is the variation distance between the probability distribution the execution of the circuit would sample from if there were no error (i.e. the ideal case) and the probability distribution it ($\Tilde{\mathcal{C}}$) actually samples from.
\end{definition}
Furthermore, verification protocols, broadly speaking, rely on a different kind of assumption to accreditation \AJ{and mirror circuit fidelity estimation}. Verification has typically used the cryptographic framework/setting of delegated computation~\cite{Ferracin_2018} (i.e. is based on two characters Alice and Bob, each with their own \emph{opposing} aims) interacting and precluding error typically\footnote{There exist verification protocols that do not assume there is no error in some particular aspect of a computation, using two \emph{non-communicating} quantum computers~\cite{v012a003, PhysRevLett.120.040501, Natarajan_2017}, but herein I focus on single-device protocols.} either in state preparation~\cite{Barz2013, Kashefi_2017, Kapourniotis2019nonadaptivefault} or measurement~\cite{ Hangleiter_2017, PhysRevLett.120.040501, Markham_2020}); accreditation \AJ{and mirror circuit fidelity estimation} have instead tended to base \AJ{their} assumptions on the experimentally-observed physics within quantum computers.

By turning from the cryptographic and trust-based assumptions of verification protocols~\cite{Gheorghiu_2018} -- from which it originates -- to physics-based assumptions, accreditation has been able to develop assumptions and then protocols~\cite{Ferracin_2019} -- using those assumptions -- to initiate a vein of research that has thus far: provided the aforementioned scalable protocol to achieve confidence in computation outputs, proven itself experimentally implementable~\cite{Ferracin_2021}, and led to the first methods for quantifying the quality of the outputs of quantum analogue simulations~\cite{jackson2023accreditation, jackson2025improvedaccreditationanaloguequantum}.

But these successes have come at a cost. Most relevantly, for this paper, accreditation has hitherto required abandoning the adversarial noise model present in many verification protocols~\cite{Gheorghiu_2018}; moving, instead, to a model where noise induces identically and independently distributed (IID) CPTP  error\footnote{Meaning in each execution of a circuit, the error is represented by identical and independent -- across multiple executions -- completely positive trace preserving maps acting on both the system and its environment.}. Returning to an adversarial error model / problem setting has proven difficult as the physics-based assumptions of accreditation protocols clash with the adversarial model of noise.
This is resolved herein by upgrading the protocol in Ref.~\cite{Ferracin_2021} to assume a limited form of adversarial noise wherein the cryptographic setting / Alice and Bob formalism~\cite{Ferracin_2018}, as used in adversarial noise models, is used but modified with a protocol (Protocol~\ref{adversarialExecutionProtocol}) for how quantum circuits are executed that both Alice and Bob participate in -- still leaving Bob to be as malicious as he likes but limiting, based on the aforementioned physics-based assumptions, his knowledge (of the circuits to be executed) and abilities (to influence the execution of circuits) by adding a new, impartial, character, Robert, who actually perfroms all quantum computations.
This upgrade is summarized in Table~\ref{tab:protocolComparison} and improves upon the IID CPTP error model, in Ref.~\cite{Ferracin_2021}, while retaining validity in all situations where it was valid before. The result is a protocol with less stringent assumptions, loosening the requirements on the physical computers used to implement it.

\begin{table}[b]
\begin{tabular}{ll}
 {\ul Ref.~\cite{Ferracin_2021}}   & {\ul This paper}          \\
 Error modelled as CPTP maps & Error modelled as CPTP maps\\
 Error is  IID & Error is adversarial (but limited) \\
 GI single-qubit gate error        &  GI single-qubit gate error
\end{tabular}
\caption{Table showcasing the upgrade presented in this paper as a comparison between the assumptions in Ref.~\cite{Ferracin_2021} and herein. Note that GI is shorthand for Gate-Independent (see Sec.~\ref{sec:ErrorModel:justify}). 
}
    \label{tab:protocolComparison}
\end{table}

This is achieved through exploiting limits placed on Bob, and enforced by Robert, (based on experimental realities -- that the probability of error varies little between executions of the same circuit, on the same hardware, in the same environment -- and inspired by single-qubit gates experiencing gate-independent (GI) error in pre-existing accreditation protocols~\cite{Ferracin_2019, Ferracin_2021}) allowing for methods to overcome the factors that prohibited an upgrade to assuming adversarial noise in previously extant accreditation protocols.

This paper proceeds, from here, with an introduction to the adversarial problem setting used throughout this paper -- in Sec.~\ref{sec:ErrorModel:presentation} -- and its justification -- in Sec.~\ref{sec:ErrorModel:justify}. I then present my main result: the adaptation of accreditation protocols to the newly-developed problem setting in Sec.~\ref{sec:LargerAccreditationProtocol}. The paper then concludes -- in Sec.~\ref{discussionSec} -- with a discussion of further possible developments.

\section{Problem Setting of This Paper}
\label{sec:ErrorModel}
As this manuscript bridges the gap between the physics-based assumption model of previous accreditation papers and the cryptographic setting that is more typical of verification protocols, I take an approach that I hope will satisfy both tribes: below I present the problem setting (summarized in Table~\ref{table:AliceAndBob}) -- where Alice, Bob, Robert and their respective objectives are introduced -- the crux of which is a protocol (Protocol~\ref{adversarialExecutionProtocol}) for how circuits are executed (which Alice and Bob both participate in but neither actually implement the computation themselves; instead, Robert -- who is fanatically and exclusively devoted to correctly performing his role in Protocol~\ref{adversarialExecutionProtocol} -- does), and then the physics of errors occurring in a circuit execution are used to support and justify the problem setting, for those who prefer a physics-based error model.

The foundation of my problem setting\footnote{Which can be seen as akin to the error model in previous accreditation protocols~\cite{Ferracin_2019, Ferracin_2021, jackson2023accreditation, jackson2025improvedaccreditationanaloguequantum}.} is the assumption that error in any operation (e.g. state preparation, gate implementation, or measurement) may be considered as the ideal/errorless operation followed or preceded by a CPTP map, as in Ref.~\cite{Ferracin_2021} and depicted in Fig.~\ref{fig:CPTPDefiningErrorExample}. 

Modeling error as a CPTP map is justified as any map from and to density matrices is a CPTP map. Hence, any erroneous operation must be a CPTP map and so any erroneous implementation of an operation may be written as the errorless/ideal operation either followed or preceded by a CPTP map -- as unitaries are CPTP maps. 

Sec.~\ref{sec:ErrorModel} continues, from here, in Sec.~\ref{sec:ErrorModel:AliceAndBob} with an introduction to the two characters of the problem setting: Alice and Bob. It explains their aims and capabilities, but that does not completely characterize the problem setting as their exact mode of interaction to achieve these aims are not yet specified. This requires an intermission, in Sec.~\ref{sec:ErrorModel:definitions}, from the presentation of the problem setting; Sec.~\ref{sec:ErrorModel:definitions} defines the concepts of redaction and CPTP lists, which are vital for when the presentation of the problem setting is completed in Sec.~\ref{sec:ErrorModel:presentation}. Sec.~\ref{sec:ErrorModel:presentation} specifies exactly how Alice and Bob interact to achieve their respective -- and competing -- aims, with the help of a new character, Robert, who is impartial and only wants to facilitate the interactions of Alice and Bob.

\subsection{Introduction to Alice, Bob, and Robert}
\label{sec:ErrorModel:AliceAndBob}
With its foundations established, I now present my problem setting -- which is an adaption of the cryptographic setting~\cite{Ferracin_2019} -- beginning with the characters of the problem setting: Alice, Bob, and Robert. 

In the adapted cryptographic setting used herein, Alice is attempting to get the result of a specific sampBQP~\cite{Aaronson2014, Lund2017,  10.5555/3135595.3135617} quantum computation\footnote{Meaning a computation to obtain a sample from a specified distribution, that can be efficiently performed on a quantum computer. Note that sampBQP includes both BQP and fBQP (the set of functions efficiently computable on a quantum computer).} of her choosing, while only having the ability to perform polynomially-bounded classical computation and initiate -- and participate in -- Protocol~\ref{adversarialExecutionProtocol}.

Bob is computationally unbounded and -- using his participation in Protocol~\ref{adversarialExecutionProtocol} when Alice invokes it -- aims to trick Alice into accepting the results of an incorrect computation, believing they are the outputs of the computation she wanted to be performed, when they are not.

Bob represents the noise in the computation (the choices he makes in Protocol~\ref{adversarialExecutionProtocol} are choosing which error get applied to the computation Alice requests) and is the adversary that ``adversarial noise" gets its name from. This personification of the noise is used as a worst-case scenario, as real noise is not actually that smart or malicious, but if a protocol allows Alice to defend against Bob -- in the adversarial problem setting -- it will also work when computations experience the less sophisticated noise that is more typical in reality.

The limits of Bob \emph{are} contained within Protocol~\ref{adversarialExecutionProtocol} but these limits are enforced and personified within the modified adversarial problem setting by Robert, who also plays a role in Protocol~\ref{adversarialExecutionProtocol} and is the one who actually performs the computations -- that Alice is provided the results of -- in the problem settings. Robert also checks that Alice and Bob are conforming to all the rules of Protocol~\ref{adversarialExecutionProtocol}, and aborts the protocol if not.

\subsection{Problem Setting Preliminaries and Definitions}
\label{sec:ErrorModel:definitions}
While the aims of Alice and Bob remain exactly the same as in the standard cryptographic setting, my specific problem setting is slightly modified from the traditional cryptographic setting. Mainly through limitations on Bob and the associated addition of a new, impartial character called Robert. 

These limitations are entirely contained in how the circuits Alice requests be executed -- and Bob tries to corrupt -- are executed, which is presented in Protocol~\ref{adversarialExecutionProtocol} in Sec.~\ref{sec:ErrorModel:presentation}. Protocol~\ref{adversarialExecutionProtocol} features computations being performed not by Alice or Bob but by some third, honest, referee (whom I call Robert, as in Ref.~\cite{StilckFranca2022gameofquantum}); with inputs (as prescribed by Protocol~\ref{adversarialExecutionProtocol}) from Alice and Bob. Alice providing the computations and Bob contributing the error. The choices of these inputs, given to Robert, are how Alice and Bob each attempt to achieve their respective goals. I.e. Alice requests computations; to Robert, in Protocol~\ref{adversarialExecutionProtocol}; and Bob chooses the errors to add to the computations during their execution, which Robert dutifully applies (assuming they conform to the requirements of Protocol~\ref{adversarialExecutionProtocol}) to the computations Alice requested.

Presenting Protocol~\ref{adversarialExecutionProtocol} first requires I define a number of concepts. These start with Def.~\ref{def:redactedGates}, Def.~\ref{def:redactionclass}, and Def.~\ref{def:sameRedactionClass}.
Which, collectively, define a concept called redaction, and its usage, which is designed to model -- in the problem setting used herein -- the idea that noise is independent of which single-qubit gates are applied in a circuit but can depend on all other aspects of a circuit. As noise is adversarial in this paper, to model this I need some way of showing the adversarial noise (i.e. Bob) the circuit but hiding the single-qubit gates from it/him. Redaction is how I achieve this. It is not a physical thing we actually can do but is part of the new cryptographic model I develop herein -- and represents single-qubit gates experiencing gate-independent error (as the gates are hidden when the CPTP maps applying the error are chosen by Bob).

\subsubsection{Redaction and Related Definitions}

I now commence the series of definitions required to adequately define redaction.
\begin{definition}
\label{def:redactedGates}
    A gate in a circuit is said to be \underline{redacted} if the gate's position in the circuit (i.e. when it is applied and to which qubits) is specified but the operator it represents (e.g. if it is a Pauli $\mathcal{X}$ gate or Pauli $\mathcal{Z}$ gate) is not. For an example of a circuit with a redacted gate, see Fig.~\ref{fig:redactionExample}.
    \begin{figure}[b]
    \centering
    \includegraphics[width=0.32\textwidth]{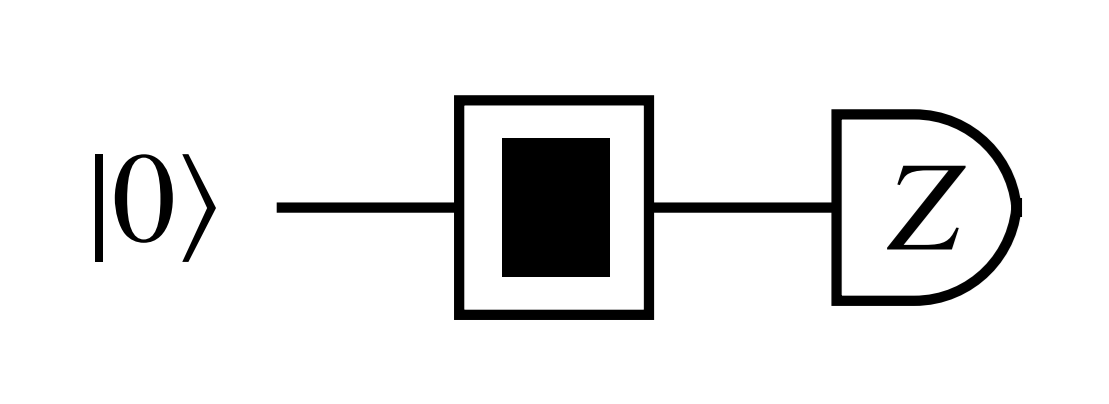}
    \caption{An example circuit with a redacted gate (the gate with a black block, in the middle of the circuit). Note that the location of the gate (in terms of when it is applied and which qubits are affected by it) is depicted but which operation the gate represents is hidden by the black block and hence is unknown to anyone viewing this redacted circuit.} \label{fig:redactionExample}
\end{figure}

Similarly, a circuit is redacted (i.e is a \underline{redacted circuit}) if all of its single-qubit gates are redacted, e.g. the circuit in Fig.~\ref{fig:redactionExample}.
\end{definition}
\begin{definition}
\label{def:redactionclass}
    Given a redacted circuit, the set of all circuits, without redactions, that the given redacted circuit possibly could be, if its redactions were removed, is called its \underline{redaction class}.
    For example, all the circuits in Fig.~\ref{fig:redactionClassExample} are in the redaction class of redacted circuit in Fig.~\ref{fig:redactionExample}.
\end{definition}

\begin{definition}
\label{def:sameRedactionClass}
Two circuits are \underline{in the same redaction class} if there exists a redacted circuit such that both the given circuits are in that redacted circuit's redaction class.

For example, all the circuits in Fig.~\ref{fig:redactionClassExample} are in the same redaction class: the redaction class of the redacted circuit in Fig.~\ref{fig:redactionExample}.
    \begin{figure}
    \centering
    \includegraphics[width=0.24\textwidth]{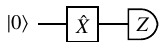}
    \includegraphics[width=0.25\textwidth]{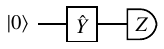}
    \includegraphics[width=0.25\textwidth]{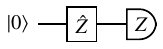}
    \caption{Example circuits within the same redaction class, where the redaction class corresponds to the redacted circuit in Fig.~\ref{fig:redactionExample}. Note that the circuit in Fig.~\ref{fig:redactionExample} could be any of the above circuits if the redaction on its single-qubit gate were removed.} \label{fig:redactionClassExample}
\end{figure}
\end{definition}

\subsubsection{CPTP Lists and Error-Related Definitions}
The second and final round of terminology required for the presentation of Protocol~\ref{adversarialExecutionProtocol} is related to how Bob specifies the CPTP maps that implement the error in a specific execution of a circuit.

When Bob wants to provide Robert with CPTP maps that implement the error in a specific execution, he gives Robert a CPTP list, as defined in Def.~\ref{def:CPTPLIst}. But Bob cannot just give Robert any CPTP list he likes, there are limits imposed by Protocol~\ref{adversarialExecutionProtocol} (and rigidly enforced by Robert) that the CPTP lists Bob provides must conform to:
\begin{enumerate}
    \item The CPTP list must match (as defined in Def.~\ref{def:fitting}) the circuit it will be applied in the execution of (basically meaning the CPTP map is capable of defining the error in the circuit execution).
    \item All CPTP lists Bob provided in a single instance of Protocol~\ref{adversarialExecutionProtocol} must be from a single \label{def:PSCL} Set of Probabilistically Similar CPTP Lists with parameter $\beta$ (defined in Def.~\ref{def:PSCL}), declared at the beginning of Protocol~\ref{adversarialExecutionProtocol}, for a $\beta \in \mathbb{R}$ known before the start of the protocol.  
\end{enumerate}
\begin{definition}
    \label{def:CPTPLIst}
    A \underline{CPTP list} is an ordered set of CPTP maps.
\end{definition}
However, a CPTP list is just a list of CPTP maps. In order to be able to describe / determine the error in a given circuit, the CPTP list must contain exactly the right number of CPTP maps and each must act on exactly the right number of qubits. In this case the CPTP list is said to \emph{fit} the circuit, as in Def.~\ref{def:fitting}.
\begin{definition}
    \label{def:fitting}
    A CPTP list \underline{fits} a circuit if it may be used to describe the error\footnote{Any resulting error is valid (from the perspective of fitting the circuit). It need not be in any way physically justified or correspond to a particular quantum computer's typical error.} in an execution of that circuit.

    A CPTP list achieves this by containing exactly one CPTP map corresponding to each location in the circuit where error may occur\footnote{These locations are: immediately after each gate, immediately after state preparation, and immediately before measurement.} and each CPTP map acts on the required number of qubits (determined by the location it corresponds to in the circuit).

    A CPTP list will \underline{fit} many different circuits and for any execution of a circuit there exists a CPTP list that both fits the circuit and accurately represents the error occurring in the circuit. I refer to this as the CPTP list \underline{determining the error} occurring in a circuit execution.
    \newline
    
    For an example of a CPTP list fitting a circuit and then determining the error in a circuit execution, see Fig.~\ref{fig:CPTPDefiningErrorExample}.
        \begin{figure}
    \centering
    \includegraphics[width=0.485\textwidth]{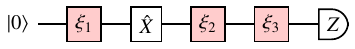}
    \caption{If $\Xi_1 = [\xi_1, \xi_2, \xi_3]$ is a CPTP list used to determine the error in an execution of the top circuit in Fig.~\ref{fig:redactionClassExample}, then Fig.~\ref{fig:CPTPDefiningErrorExample} depicts the circuit that is actually executed. $\xi_1$ describes the error due to state preparation, $\xi_2$ describes the error due to the single-qubit gate, and $\xi_3$ describes the error due to measurement.\\
    Note that notation is slightly abused to display the CPTP maps from $\Xi$ as gates (depicting error) in a circuit. To remedy this abuse slightly CPTP maps are highlighted in red to denote they are not unitaries but are CPTP maps. \label{fig:CPTPDefiningErrorExample}}
\end{figure}
\end{definition}
\begin{Note}
    \label{Note:redactionClassfits}
    For any set of circuits within the same redaction class, if a given CPTP list fits one of them, it fits all of them.
\end{Note}
    An important feature about CPTP lists and the circuits they fit is mentioned in Lemma~\ref{lem:errorDeteredByCPTPList}. Proof is omitted as Lemma~\ref{lem:errorDeteredByCPTPList} follows from the above discussion.
\begin{lemma}
    \label{lem:errorDeteredByCPTPList}
    The error in a circuit execution is entirely determined by the CPTP list used to describe its error\footnote{And -- potentially -- the outcomes of any stochastic processes in the application of the error it describes.}.
\end{lemma}
A key metric of error occurring in a circuit execution is defined in Def.~\ref{def:probOfError}, which allows for a useful quantification of the effect of the error.
\begin{definition}
\label{def:probOfError}
    For any execution of a circuit afflicted by only stochastic Pauli error, the \underline{probability of error} of that execution is the probability that \emph{any} Pauli gate is erroneously applied at \emph{any} point in the execution by an error channel. I.e. it is the probability that none of the stochastic Pauli error channels in the circuit execution apply a Pauli gate.
    
    For example, if, for all density matrices, $\rho$:
    \begin{align}
        \xi_1 (\rho) &= \dfrac{1}{2}\rho + \dfrac{1}{2} \hat{X} \rho \hat{X}^{\dagger}\\
        \xi_2 (\rho) &= \dfrac{1 - 0.01}{2}\rho + \dfrac{1 + 0.01}{2} \hat{X} \rho \hat{X}^{\dagger}\\
        \xi_3 (\rho) &= \dfrac{1 + 0.01}{2}\rho + \dfrac{1 - 0.01}{2} \hat{X} \rho \hat{X}^{\dagger},
    \end{align}
    then the \underline{probability of error} in the execution of a circuit where these channels model the error (e.g. Fig.~\ref{fig:CPTPDefiningErrorExample}) is the probability that any of $\xi_1$, $\xi_2$, or $\xi_3$ apply error to the circuit execution ($1 - \frac{1}{2} \frac{1 - 0.01}{2} \frac{1 + 0.01}{2} \approx 0.875$).
\end{definition}
Having defined the probability of error in an execution of a circuit, which is a consequence of the CPTP list determining the error in that execution, I can define a formal concept, in Def.~\ref{def:PSCL}, wherein a set of CPTP maps that may govern error in circuit executions all have a similar probability of error.
\begin{definition}
    \label{def:PSCL}
    A \underline{Set of Probabilistically Similar CPTP Lists} with parameter $\beta \in \mathbb{R}$ (abbreviated as a \underline{SPSCL$_\beta$}) is a set of CPTP lists, that all fit the same set of circuits (as in Def.~\ref{def:fitting}), such that, for any CPTP list in the set, the probability of error in any circuit execution it determines the error for, once it is twirled to stochastic Pauli error, is within the interval:
    \begin{align}
        \label{Eqn:ErrorModelsConditions}
        \big[ P_0 \big( 1 - \beta\big), P_0 \big( 1 + \beta \big) \big],
    \end{align}
    for some $P_0 \in [0,1]$ (such that $P_0 \big( 1 + \beta \big) \leq 4/ 5$).
    \newline

    For an example of a SPSCL$_\beta$, consider two CPTP lists: $\Xi_2 = [\xi_1]$ and $\Xi_3 = [\xi_2]$ (where $\xi_1$ and $\xi_2$ are as in Def.~\ref{def:probOfError}).
    Then the set $\{ \Xi_2, \Xi_3\}$ can be considered as an SPSCL$_\beta$ for any $\beta > 0.01$ (where the value of $P_0$ in this SPSCL$_\beta$ is $0.5$).
\end{definition}

\subsubsection{Relating CPTP Lists and Redaction Classes}

With both SPSCL$_\beta$ and redaction classes defined, they may be related to each other via the below Lemma~\ref{lem:definingMatching}. As will be seen, for each use of Protocol~\ref{adversarialExecutionProtocol}, Alice chooses a redaction class (implicitly, as all circuits must come from the same redaction class), and Bob chooses an SPSCL$_\beta$ to determine the error in those executions. Hence, SPSCL$_\beta$ and redaction classes are in some sense duals of each other. This duality is further expounded apon by Lemma~\ref{lem:definingMatching}.
\begin{lemma}
    \label{lem:definingMatching}
    For any given redaction class and any given SPSCL$_\beta$ $\bold{either}$:\\
    $\bold{1.}$ Every CPTP list in the SPSCL$_\beta$ fits every circuit in the redaction class\\
    $\bold{or}$\\
    $\bold{2.}$ No CPTP list in the SPSCL$_\beta$ fits any circuit in the redaction class.\\ \( \\ \)
    I refer to the first of the above options as the redaction class and the SPSCL$_\beta$ \underline{matching} (but also allow the term to apply if the redaction class is replaced by any subset of itself).  
\end{lemma}
\begin{proof}
    As in Note~\ref{Note:redactionClassfits}, if a given CPTP list fits one circuit in a redaction class, it fits every circuit in that redaction class.
    As in Def.~\ref{def:PSCL}, if one CPTP list in a given SPSCL$_\beta$ fits a specific circuit then every CPTP list in that SPSCL$_\beta$ fits that circuit.
    Lemma~\ref{lem:definingMatching} follows from combining these two facts.
\end{proof}

\subsection{Problem Setting Presentation}
\label{sec:ErrorModel:presentation}
I can now present the protocol by which sets of circuits, in my problem setting, are executed, in Protocol~\ref{adversarialExecutionProtocol}. Both Alice and Bob are required to participate in Protocol~\ref{adversarialExecutionProtocol}, when initiated, and both must abide by all requirements of it at a. These requirements are enforced by another character, Robert, who actually implements the quantum computations -- in Protocol~\ref{adversarialExecutionProtocol} -- according to Alice and Bob's combined instructions. However, Protocol~\ref{adversarialExecutionProtocol} still leaves room for Bob to attempt to trick Alice.

\begin{figure}
    \centering
\begin{algorithm}[H]

\SetAlgorithmName{Protocol}{protocol}{List of Protocols}

\begin{enumerate}
    \item Alice provides a set of circuits to execute, $\mathcal{S}$, all within the same redaction class, to Robert.
    \item Bob receives the full details of $\mathcal{S}$, from Robert, but with every circuit in $\mathcal{S}$ redacted.
    \item Bob chooses an SPSCL$_{\beta}$, denoted $\Xi$, that matches $\mathcal{S}$ and tells Robert it, where $\beta \in \mathbb{R}^+$ is known to both Alice and Bob in advance.
    \item For each circuit, $\mathcal{C}$, in $\mathcal{S}$:
    \begin{enumerate}
        \item Bob chooses a CPTP list, $\xi$, from $\Xi$ and gives it to Robert.
        \item $\mathcal{C}$ is executed, by Robert, with $\xi$ determining the error occurring during the execution.
        \item Robert gives both Alice and Bob the measurement outcomes of the execution of $\mathcal{C}$.
    \end{enumerate}
\end{enumerate}

\caption{Formal Adversarial Model: How Sets of Circuits are Executed
 \label{adversarialExecutionProtocol}}
\end{algorithm}
\end{figure}
Alice is free to initiate Protocol~\ref{adversarialExecutionProtocol} whenever and as often as she likes; Bob is free to act how he likes to achieve his stated malicious aims but must execute Protocol~\ref{adversarialExecutionProtocol} as prescribed -- but is free to act nefariously within the permitted bounds -- when it is invoked. Robert checks Bob's nefarious choices are within the rules of Protocol~\ref{adversarialExecutionProtocol} and executes circuits. 

The aims of Alice and Bob remain the same as in the standard cryptographic setting and as discussed in Sec.~\ref{sec:ErrorModel:AliceAndBob}: Alice is attempting to get Robert (via Protocol~\ref{adversarialExecutionProtocol}) to perform a quantum computation and provide her with the true results, while Bob is trying to trick Alice by adding error to Alice's requested circuit (when Robert executes it in Protocol~\ref{adversarialExecutionProtocol}) so that Robert provides Alice with the results of an incorrect computation and Alice believes that the computation was performed without error. I.e. that the results Alice receives are the outputs of the computation she wanted to be performed when they are not (due to the error Bob adds). However, while trying to achieve these aims, Alice and Bob can only communicate via Protocol~\ref{adversarialExecutionProtocol} -- with Robert.
Robert's aims are entirely neutral, he has no preference on if Alice or Bob successfully achieves their aim. He just wants to do his duty and perform Protocol~\ref{adversarialExecutionProtocol} correctly.  
The problem setting of this paper is summarized in Table~\ref{table:AliceAndBob}.
\begin{table*}[]
\begin{tabular}{l|l|l|l}
                           & \textbf{Alice}                                                              & \textbf{Bob}
                           & 
                           \textbf{Robert}
                           \\ \hline
Aim                        & Obtain the results of a specific computation        & Alice to accept incorrect results & Implement Protocol~\ref{adversarialExecutionProtocol} as requested \\
Computational Capabilities & Polynomial time classical computation                              & Unbounded  & sampBQP computations                                                                \\
Additional Capabilities    & May initiate Protocol~\ref{adversarialExecutionProtocol} at any time & None & None                                                                      \\
Allowed Communication       & None, aside from via Protocol~\ref{adversarialExecutionProtocol}   & None, aside from via Protocol~\ref{adversarialExecutionProtocol} 
& None, aside from as in Protocol~\ref{adversarialExecutionProtocol} 
\end{tabular}
\caption{\label{table:AliceAndBob} Table summarizing the problem setting of this paper (i.e. the adapted cryptographic setting) in terms of the aims, computational abilities, additional capabilities, and allowed communication of its only characters: Alice, Bob, and Robert.}
\end{table*}

While the bulk of the rest of this paper focuses on constructing the accreditation protocol, without much regard for Alice, it culminates in Theorem~\ref{thm:AccProtocolIsCorrect} showing that Alice's problem is solved by the protocol presented herein. I.e. in the problem setting from this section, the protocol presented in this paper allows Alice to get Bob to perform a quantum computation for her and have confidence in the results she receives.
\subsection{Physical Justification of the Problem Setting}
\label{sec:ErrorModel:justify}
The two main limitations on the adversarial error / Bob, both following from Protocol~\ref{adversarialExecutionProtocol}, requiring justification are:
\begin{enumerate}
    \item The single-qubit gates are redacted when the circuits in $\mathcal{S}$ are shown to Bob (in step 2 of Protocol~\ref{adversarialExecutionProtocol}).
    \item For all circuit executions within a single use of Protocol~\ref{adversarialExecutionProtocol}, the different CPTP lists Bob uses to determine the error in each execution are all within a single SPSCL$_{\beta}$, for a known $\beta \in \mathbb{R}^+$.
\end{enumerate}
Both of these limitations are applied (in the case of the first) or enforced (in the case of the second), in Protocol~\ref{adversarialExecutionProtocol}, by Robert.

The first limitation corresponds to single-qubit gates experiencing gate-independent error. This is a standard assumption the the pre-existing accreditation protocols (in Refs.~\cite{Ferracin_2019, Ferracin_2021, jackson2023accreditation}) and follows from single-qubit gates typically being the least error-prone components of a quantum computer~\cite{Arute2019, Wright2019} (which has held true over time~\cite{doi:10.1126/science.1145699, PhysRevLett.113.220501}, and can be most clearly seen in Ref~\cite[Fig.~5]{9355264}): the error is so small that the error in different single-qubit gates cannot differ much.
This is a ``standard [assumption] in the literature on noise
characterisation and mitigation"~\cite{ferracin2022efficiently} and has seen extensive use in theoretical work~\cite{PhysRevA.77.012307, PhysRevLett.106.180504, PhysRevLett.106.230501, PhysRevLett.109.070504,PhysRevA.87.062119, PhysRevLett.114.140505, PhysRevA.92.060302, Erhard2019, Harper2020, Dahlhauser_2021, ferracin2022efficiently}.

The second limitation is a weakening of the IID assumption in previous accreditation protocols.
The intuition behind it is that the same hardware executing very similar circuits in quick succession -- as happens in Protocol~\ref{adversarialExecutionProtocol} -- will experience similar error in each circuit execution as:
\begin{enumerate}
\item The hardware executing the similar circuits is the same in each execution and any error-inducing aspects of the hardware, or surrounding environment, are unlikely to change much in the short time between executions.
\item The effect of any aspect of the circuits being executed that may change the error is minimized, as the circuits are very similar.
\end{enumerate}
Although no paper has, to my knowledge, sought to directly validate this limitation, it can be justified experimentally:
\begin{enumerate}
\item Ref.~\cite[Figure.~5]{Ferracin_2021} ran the accreditation protocol developed therein many times, producing a probability of error in the trap circuit executions for each use of the protocol. The set of error probabilities generated by this can be seen to vary little across many uses of the protocol.
\item Ref.~\cite[Fig.~6, Fig.~7, and Fig.~8]{tannu2018case} investigated the error rates in single-qubit and two-qubit gates. It found these error rates are very rarely far from their average. Note that the variation shown in this paper is over a much longer time-span (days) than this limitation requires (seconds).
\item Ref.~\cite[Fig.~9]{9355264} examined NISQ computers and plotted the error rates of different qubits in each device. Fig.~9 shows the error bars on the error rate to be -- with some notable exceptions -- small fractions of the error rate and varying much more across qubits than for a fixed qubit over time.
\item Ref.~\cite[Fig.~1(b)]{10.1145/3297858.3304075} looked at the error rates of two-qubit gates on differing pairs of qubits, it shows that (again, over a time-span of days) for a specific pair of qubits, the days when the error rate deviates far from its average are rare, with more extreme deviations being rarer.
\end{enumerate}
Finally, this limitation can be seen as an aim of quantum computer hardware engineering: as quantum computers improve and their actual outputs approach the errorless outputs, with decreasing variance, the typical difference -- by any measure -- between the error contained in two executions of similar circuits will tend to zero.
 
To my knowledge, I am the first to explicitly state this second limitation. However, I note that it is already implicitly accepted in the community by the acceptance of randomized benchmarking~\cite{PhysRevA.77.012307, Dankert_2009, PhysRevLett.123.060501, PRXQuantum.3.020357} as a meaningful measure: if even the same circuit repeated multiple times, in very quick succession, produces wildly varying probabilities of error then the variation of error probabilities implies that previous measures mean almost nothing for future computations or the quality of a quantum device over a meaningful time-span. This is not the case and Ref.~\cite[Fig.~6]{10.1145/3297858.3304075} shows that, over a span of days, the results of randomized benchmarking do not vary over a very large range.

\section{Adversarial Accreditation Protocol}
\label{sec:LargerAccreditationProtocol}
Sec.~\ref{sec:LargerAccreditationProtocol} is dedicated to resolving the problem Alice faces in the problem setting established in Sec.~\ref{sec:ErrorModel:presentation}. This is equivalent to upgrading the accreditation protocol in Ref.~\cite{Ferracin_2021} that assumed error is CPTP and IID (and that single-qubit gates suffer only gate-independent error) to one that works in the problem setting described in Sec.~\ref{sec:ErrorModel:presentation}.

Sec.~\ref{sec:LargerAccreditationProtocol} begins -- in Sec.~\ref{sec:TrapsAndTargets} -- with a presentation of the trap and target circuits I intend to use in the new accreditation protocol, which are very similar to those in Ref.~\cite{Ferracin_2021}. As a trap-based verification protocol, the accreditation protocol presented in Sec.~\ref{sec:LargerAccreditationProtocol} needs these trap circuit executions to be executed alongside the target circuit (in the same single use of Protocol~\ref{adversarialExecutionProtocol} and hence experiencing comparable error by the assumptions implicit in the problem setting defined in Sec.~\ref{sec:ErrorModel}) and give a measure of the quality of the execution of the target circuit.  
The usage of these trap and target circuits to produce an accreditation protocol is them detailed in Sec.~\ref{sec:DigitalAccProtocol} (and more formally presented in Algorithm~\ref{StandardAccAlg}).

\subsection{Trap and Target Circuits}
\label{sec:TrapsAndTargets}
In this paper, I do not propose to develop new trap circuits or target circuits. In fact, I would prefer to make minimal changes to the trap and target circuits in Ref.~\cite{Ferracin_2021}. I will also not regurgitate the exact designs of the trap and target circuits in Ref.~\cite{Ferracin_2021} and instead throughout this paper will assume that I have two efficient classical algorithms, $P_{\textit{targ}}$ and $P_{\textit{trap}}$, that, if given any quantum circuit as input, return a random\footnote{I.e. $P_{\textit{trap}}$ and $P_{\textit{targ}}$ will choose random gates used to apply the Pauli twirls and probabilistic error detection (via Hadamard gates), so will return a slightly different circuit each time.} corresponding target and matching trap, respectively, of the protocol in Ref.~\cite{Ferracin_2021}. I briefly note the important features of these trap and target circuits (that will be inherited wherever I use trap or target circuits herein):
\begin{enumerate}
    \item In target circuits and trap circuits all error occurring is twirled (via Pauli twirls) to stochastic Pauli error and is thereafter considered as such. 
    \item If no error occurs in a trap, it gives a specific output, $m$, and if error does occur the trap does \emph{not} give the output $m$ with probability at least $k \in [0,1]$.
    \item Target circuits and trap circuits differ only in their single-qubit gates so they are all in the same redaction class.
\end{enumerate}
The above assumptions have mentioned twirling CPTP error to stochastic Pauli error, which I define formally in Def.~\ref{def:twirls}. This is an important step in all accreditation protocols, as it reduces the error to a known, more easily quantified form.
\begin{definition}
    \label{def:twirls}
    CPTP error within a quantum circuit is said to be \underline{twirled} to stochastic Pauli error~\cite{Wallman_2016, PhysRevX.11.041039, Jain_2023} if, via the addition of only single-qubit Pauli gates to the circuit, the error is effectively transformed to stochastic Pauli error, without otherwise affecting the outputs of the circuit (e.g. it does not affect the outputs of the errorless case).
    Likewise, gates are said to be \underline{twirled} if any error occurring in them is twirled to stochastic Pauli error.
\end{definition}
As herein noise is considered to be adversarial, the accreditation protocol of this paper does require a single slight modification of the trap circuits and target circuits: Ref.~\cite{Ferracin_2021} did not consider ``hiding'' the measurement outcomes of a circuit as there was no adversary to ``hide'' them from (according to its error model). 
With adversarial noise (i.e. the problem setting of Sec.~\ref{sec:ErrorModel:presentation}), this becomes necessary as otherwise Bob may be able to identify which circuits are trap and target circuits, respectively, based on these outcomes, or base future error on the measurement outcomes of previous circuits (removing the independence of the error in different circuits). The required hiding / encrypting of measurement outcomes is achieved via Algorithm~\ref{ModifiedTrapGenAlg} which; using the classical algorithms for generating the trap and target circuit of Ref.~\cite{Ferracin_2021}, $P_{\textit{trap}}$ and $P_{\textit{targ}}$ respectively; acts as a classical algorithm to generate trap and target circuits similar to those of Ref.~\cite{Ferracin_2021} but with the outputs quantum-securely encrypted and completely unrecoverable without the key.
\begin{figure}[h!]
    \centering
\begin{algorithm}[H]
$\bold{Input:}$ \\
$\bullet$ A circuit, $\mathcal{C}$, to generate trap circuits or a target circuit for.\\
$\bullet$ Two algorithms, $P_{\textit{trap}}$ and $P_{\textit{targ}}$, that generate trap and target circuits, respectively, without hidden outputs.\\
$\bullet$ A Boolean, labeled \emph{isTarget}, denoting if a trap or target is to be generated.\\
    \begin{enumerate}
        \item If \emph{isTarget} == True:
            \begin{enumerate}
                \item $\mathcal{C}'$ = $P_{\textit{targ}} \big( \mathcal{C} \big)$.
            \end{enumerate}
            Else:
            \begin{enumerate}
                \item $\mathcal{C}'$ = $P_{\textit{trap}} \big( \mathcal{C} \big)$.
            \end{enumerate}

            \item Generate a random bit string with the same length as the number of measurements in $\mathcal{C}'$, referred to as the \emph{key}.
            
            \item For measurement, \emph{M}, in $\mathcal{C}'$:
            \begin{enumerate}
                \item Calculate the single-qubit unitary, $\mathcal{U}_M$, that, if applied immediately before measurement, \emph{M}, would flip the outcome.
                \item If ( the bit in the \emph{key} corresponding to measurement \emph{M} ) == 1:
                \begin{enumerate}
                    \item Add $\mathcal{U}_M$ to $\mathcal{C}'$ immediately before measurement \emph{M}.
                \end{enumerate}
            \end{enumerate}
    \end{enumerate}
$\bold{Return}:$ $\mathcal{C}'$ and the \emph{key}.
\caption{Generating Trap and Target Circuits with Hidden Outputs
 \label{ModifiedTrapGenAlg}}
\end{algorithm}
\end{figure}

I note that the required (by Algorithm~\ref{ModifiedTrapGenAlg}) single-qubit unitaries, $\mathcal{U}_M$, will exist for any single-qubit measurement and that the outcome of the circuit -- if the single-qubit gates of the circuit are not known, as is the case for Bob as they are redacted for him -- is irretrievable from the measurement outcomes without the key.
This hiding of the circuit outputs is information-theoretically secure (i.e. has perfect security)~\cite{6769090}, assuming the single-qubit gates are redacted, and comparable to the one-time pads in universal blind quantum computing~\cite{Broadbent_2009}. Henceforth, I will denote $P_{\textit{targ}}$ and $P_{\textit{targ}}$, with the changes implemented by Algorithm~\ref{ModifiedTrapGenAlg} to hide their outputs, by $P_{\textit{targ}}'$ and $P_{\textit{targ}}'$ respectively.
I note that if you have the key, the output of the circuit can be easily recovered by XOR-ing each measurement outcome with the corresponding bit in the key. I will additionally use ``outputs" to refer to the results of traps and targets \emph{after} the key has been used to undo the effect of hiding the outputs from Bob.

\subsection{Presentation of the Upgraded Accreditation Protocol}
\label{sec:DigitalAccProtocol}
With trap and target circuits (with securely encrypted measurement outcomes) established, I can define the full accreditation protocol.
In line with Sec.~\ref{sec:TrapsAndTargets}, for Sec.~\ref{sec:DigitalAccProtocol}, trap circuits and target circuits will be treated as black boxes and I will only refer to their construction as being performed by the polynomial-time classical algorithms, $P_{\textit{targ}}'$ and $P_{\textit{targ}}'$. Their only relevant properties will be that trap circuit executions detect any error with probability at least $k \in (0,1]$, all error in a trap or target simulation is  effectively reduced to stochastic Pauli error, and Bob cannot tell the difference between trap and target simulations due to them only differing in their single-qubit gates (which makes them indistinguishable to Bob, during Protocol~\ref{adversarialExecutionProtocol}, due to the redaction of all single-qubit gates, by Robert, before they are shown to Bob).

However, before the full protocol can be presented, in Protocol~\ref{StandardAccAlg}, I must establish the statistical foundations of the new protocol -- in Sec.~\ref{sec:statisticalBasisofAcc} -- and the core mechanics of the protocol -- in Sec.~\ref{sec:coreMechanicsOfAcc}. These statistical methods will be used to evaluate the $P_0 \in [0,1]$ that defines Bob's particular choice of SPSCL$_{\beta}$ (as in Def.~\ref{def:PSCL}), using multiple trap executions ( all within a single use of Protocol~\ref{adversarialExecutionProtocol}).
If a target circuit is then executed within that same single use of Protocol~\ref{adversarialExecutionProtocol} as the trap circuit executions, this allows the probability of error (as in Def.~\ref{def:probOfError}) of the execution of the target circuit to be bounded upper (using that all circuit executions in a single use of Protocol~\ref{adversarialExecutionProtocol} have error within a single SPSCL$_{\beta}$). Due to the argument in  Ref.~\cite[Appendix Sec. 1]{Ferracin_2021}, this, in turn, upper bounds the ideal-actual variation distance (as defined in Def.~\ref{def:idealActualVarDist}) of the target circuit execution.

\subsubsection{Statistical Basis of the New Accreditation Protocol}
\label{sec:statisticalBasisofAcc}
Before presenting the accreditation protocol, it is first useful to present a purely statistical lemma (Lemma~\ref{BoundingAllinInterval}) that will later enable the accreditation protocol. 

For any set of real values, $\mathcal{R} = \big \{ r_j \in \mathbb{R}^+ \big\}^{\vert \mathcal{R} \vert}_{j = 1}$, let $\mathrm{Avg} \big( \mathcal{R} \big)$ denote the average of $\mathcal{R}$.
 Later the set $\mathcal{R}$ will denote set of the respective probabilities of error in each trap circuit execution within a single use of Protocol~\ref{adversarialExecutionProtocol} but for Lemma~\ref{BoundingAllinInterval} $\mathcal{R}$ is just considered to contain positive real numbers with no meaning attached to them.
\begin{lemma}
    \label{BoundingAllinInterval}
    Given a set of positive real values, $\mathcal{R} \subset [(1 - \beta)P_0, (1 + \beta)P_0]$ (for some $P_0 \in \mathbb{R}^+$); if $\beta \in [0,1)$ is known, $\forall y \in [(1 - \beta)P_0, (1 + \beta)P_0]$,
    \begin{align}
        y \leq \frac{1 + \beta}{1 - \beta} \mathrm{Avg}(\mathcal{R}).
    \end{align}
\end{lemma}
\begin{proof}
Assume $y$ takes the max possible value, $(1 + \beta)P_0$, and $\mathrm{Avg}(\mathcal{R})$ takes the min possible value, $(1 - \beta)P_0$. The claimed inequality can then be checked to hold.
\end{proof}
\subsubsection{Core Mechanics of the New Accreditation Protocol}
\label{sec:coreMechanicsOfAcc}
Before the formal presentation of our new accreditation protocol, in Sec.~\ref{sec:FormProtforAcc}, the core mechanics of the accreditation protocol are presented. This takes the form of Lemma~\ref{thm:ifIndependentBound}.
\begin{lemma}
    \label{thm:ifIndependentBound}
    Given an efficient classical algorithm, $P_{\textit{trap}}'$, for generating trap circuits; assuming that:
    \begin{enumerate}
        \item the probability of error in each execution is independent of the outcomes of all preceding executions
        \item any trap detects any specific error, by outputting specific measurement outcomes, with probability at least $k \in [0,1]$
        \item all trap circuits are executed via a single use of Protocol~\ref{adversarialExecutionProtocol}
    \end{enumerate}
    the probability of error of any execution of a circuit where all error is Pauli twirled, during the same single use of Protocol~\ref{adversarialExecutionProtocol} as the trap circuit executions can be upper bounded by:
    \begin{align}
        \label{eqn:LemmaBoundPresent}
        \dfrac{ 1 + \beta }{k (1 - \beta)}  \big( \bar{v} + \theta \big),
    \end{align}
    using $N_{\mathrm{Tr}} = \bigg \lceil \dfrac{2}{\theta^2} \ln{\bigg( \dfrac{2}{1-\alpha} \bigg)} \bigg \rceil + 1$ trap circuit executions, where:\\
    $\bullet$ $\bar{v} \in \mathbb{Q}$ is the fraction of trap circuit executions giving an incorrect measurement outcome,\\
    $\bullet$ $\theta \in \mathbb{R}^+$ may be chosen arbitrarily,\\
    $\bullet$ $\beta \in \mathbb{R}^+$ is as in Protocol~\ref{adversarialExecutionProtocol},\\
    $\bullet$ $\alpha \in [0,1]$ is the confidence required of the bound in Eqn.~\ref{eqn:LemmaBoundPresent}.
\end{lemma}
\begin{proof}
    Let $\mathcal{R}$ be an ordered set where the $j$th element is the probability of error occurring in the $j$th execution of a set of trap executions, interspersed among any number of other circuits, in a singular use of Protocol~\ref{adversarialExecutionProtocol}.
    
    This implies that the probability of error in any specific trap (as generated by $P_{\textit{trap}}'$) -- and in any circuit executed within the same use of Protocol~\ref{adversarialExecutionProtocol} as those trap circuit executions (such as target circuits generated by $P_{\textit{targ}}'$) -- is within an interval that may be written as:
    \begin{align}
        [ P_0 (1 - \beta), P_0 (1 + \beta)],
    \end{align}
    for some $\beta, P_0 \in \mathbb{R}^+$.
    Therefore, due to Lemma~\ref{BoundingAllinInterval}, for any circuit, labeled circuit $j$, executed in the aforementioned single use of Protocol~\ref{adversarialExecutionProtocol}, along the trap circuit executions, the probability of error occurring in its execution, $p_j$, is bounded by:
    \begin{align}
        p_j \leq \frac{1 + \beta}{1 - \beta} \mathrm{Avg}(\mathcal{R}).
    \end{align}
    The probability a specific trap returns detects that an error occurs by returning an incorrect measurement outcome, given error occurs, is (as assumed in the lemma statement) lower bounded by $k \in [0,1]$. 
    
    Therefore, by also using the independence between circuit executions of whether error occurs, $\mathrm{Avg}(\mathcal{R})$ can be approximated, by checking if the trap circuit executions detect error, -- to within additive error, $\theta$, with confidence $\alpha$ -- using Hoeffding's inequality~\cite{doi:10.1080/01621459.1963.10500830}. This only requires that $\vert \mathcal{R} \vert \geq N_{\mathrm{Tr}} = \bigg \lceil \dfrac{2}{\theta^2} \ln{\bigg( \dfrac{2}{1-\alpha} \bigg)} \bigg \rceil + 1$, to provide enough trap circuit executions that the approximation -- acquired by sampling -- of $\mathrm{Avg}(\mathcal{R})$ is within the required error, $\theta$. 
    
    Let $\Bar{\nu}$ denote the experimentally obtained approximation of the probability a uniformly selected (from $\mathcal{R}$) trap returns an output implying error has occurred (i.e. the fraction of trap circuit executions returning an ``incorrect" measurement outcome in a single execution of Protocol~\ref{adversarialExecutionProtocol}).

    This approximation of $\mathrm{Avg}(\mathcal{R})$ is simply $\Bar{\nu}$ divided by $k$ (to account for the cases where error occurs but is not detected).
    Then, with confidence, $\alpha$, the probability of error occurring in the execution of the circuit with index $j$ is bounded as:
    \begin{align}
        p_j \leq \dfrac{ 1 + \beta }{k(1- \beta)} \big( \bar{v} + \theta \big).
    \end{align}
\end{proof}
\subsubsection{Formal Presentation of the Accreditation Protocol and Proof of its Correctness}
\label{sec:FormProtforAcc}
The core components of my accreditation protocol have now been constructed and presented -- in Sec.~\ref{sec:statisticalBasisofAcc} and Sec.~\ref{sec:coreMechanicsOfAcc} -- so I can now present my accreditation protocol formally, in Algorithm~\ref{StandardAccAlg}. 
This algorithm is presented as would be used by Alice, presupposing the problem setting in Sec.~\ref{sec:ErrorModel:justify}.
\begin{figure}[h!]
    \centering
\begin{algorithm}[H]

\SetAlgorithmName{Protocol}{protocol}{List of Protocols}

$\bold{Input:}$ \\
$\bullet$ A circuit, $\mathcal{C}$.\\
$\bullet$ A required accuracy of the bound to output, $\theta$.\\
$\bullet$ A required confidence in the above accuracy, $\alpha$.\\
$\bullet$ The relevant $P'_{\textit{targ}}$, $P'_{\textit{trap}}$, $m$, and $k$.\\
$\bullet$ The value of $\beta$.\\

    \begin{enumerate}
    \item Calculate $N_l = \dfrac{2}{\theta^2} \ln{\bigg( \dfrac{2}{1 - \alpha} \bigg)} + 1$.
    \item Choose a random integer $N_{tt}$ less than $10 \cdot N_l$.
    \item $\{ \mathcal{C}_j \}_{j = 1}^{N_l + N_{tt}}$ = $N_l  + N_{tt}$ circuits generated using $P'_{\textit{trap}}$ on input $\mathcal{C}$.
    \item $\mathcal{C}^{\prime}$ = circuit generated using $P'_{\textit{targ}}$  on input $\mathcal{C}$.
    \item Execute every circuit in $\{ \mathcal{C}_j \}_{j = 1}^{N_l} \cup \{ \mathcal{C}^{\prime} \}$ in a random order (i.e. disregarding the indexing used herein) using Protocol~\ref{adversarialExecutionProtocol}.
    \item $\mathrm{TargResult}$ = output of executing $\mathcal{C}^{\prime}$
    \item $\mathrm{TrapResult}$ = fraction of $\{ \Tilde{\mathcal{C}}_j \}_{j = 1}^{N_l}$ that does not output $m$.
    \end{enumerate}
$\bold{Return}:$ $\mathrm{TargResult}$, $\frac{1 + \beta}{k (1 - \beta)} \big( \mathrm{TrapResult} + \theta \big)$.
\caption{Formal Accreditation Protocol
 \label{StandardAccAlg}}
\end{algorithm}
\end{figure}
The correct functioning of Algorithm~\ref{StandardAccAlg}, as an accreditation protocol, in the adapted/modified problem setting is proven in Theorem~\ref{thm:AccProtocolIsCorrect}.
\begin{theorem}
    \label{thm:AccProtocolIsCorrect}
    In the adapted cryptographic setting (as defined in Sec.~\ref{sec:ErrorModel}); given efficient classical algorithms, $P'_{\textit{targ}}$ and $P'_{\textit{trap}}$, that generate trap and target circuits, respectively (as in Sec.~\ref{sec:TrapsAndTargets}); Protocol~\ref{StandardAccAlg} allows Alice to execute a circuit that is equivalent (in output distribution, when no error occurs) to any given sampBQP circuit, $\mathcal{C}$, and obtain a bound on the ideal-actual variation distance of that execution (to within arbitrary accuracy $\theta \in \mathbb{R}^+$, with arbitrary confidence, $\alpha \in [0,1]$).
\end{theorem}
\begin{proof}
All circuits executed in Protocol~\ref{StandardAccAlg} are executed within a single use of Protocol~\ref{adversarialExecutionProtocol} (in step 5 of Protocol~\ref{StandardAccAlg}). Therefore, if each probability of error in a given trap is independent, then Lemma~\ref{thm:ifIndependentBound} implies that, as the target circuit is a circuit executed alongside  $N_{\mathrm{Tr}} = \bigg \lceil \dfrac{2}{\theta^2} \ln{\bigg( \dfrac{2}{1-\alpha} \bigg)} \bigg \rceil + 1$ trap circuit executions in a single use of Protocol~\ref{adversarialExecutionProtocol}, the probability of error in the target circuit's execution, $p_{t}$, is upper bounded as:
\begin{align}
    \label{eqn:boundOnTrapErrorFirstAppearanceinTheorem}
    p_{t} \leq \bigg( \dfrac{1 + \beta}{k( 1 - \beta)} \bigg) \big( \bar{v} + \theta \big),
\end{align}
where $\bar{v}$ is the fraction of trap circuit executions returning an incorrect measurement result and $k \in [0,1]$ is as defined in  Sec.~\ref{sec:TrapsAndTargets}.

As mentioned, for Eqn.~\ref{eqn:boundOnTrapErrorFirstAppearanceinTheorem} to hold, it is required that the probability of error in each trap is independent of the outcomes of every prior trap. This is not initially guaranteed as Bob is free to choose probabilities based on previous measurement outcomes.

However, due to the changes made to each trap and target in Algorithm~\ref{ModifiedTrapGenAlg}, to hide the outputs of each trap from Bob with perfect security, Bob does not know the result of any trap or target and, as the limitation on Bob that means the probability of error due to the error he applies cannot be deterministic, Bob -- as he is not told about the outcomes of the random processes in the error he applies -- is unaware of what errors have been applied in each circuit execution.

Additionally, due to the $N_{tt} \in \mathbb{N}_0$ extra trap circuits executed during Protocol~\ref{adversarialExecutionProtocol}, Bob does not know what he has done to the previous circuits as he has no idea which are ``decoy" circuit executions and will be discarded (so he cannot know what error he has applied to the ``true" circuit executions). In fact, as Bob is not told $\theta, \alpha$, or $N_{tt}$, he cannot even calculate the various probabilities that he has applied each possible combination of errors -- or lack thereof -- to the ``true" trap circuit executions.

Therefore, Bob cannot make any of his choices based on the results of any prior trap or target execution or the error that occurs in them. Hence, each trap execution can be considered to be independent i.e. the probability of error in each trap or target is the result of Bob's choices but must remain completely independent of the results of any prior trap or target.

As per Protocol~\ref{adversarialExecutionProtocol}, when Bob is shown the set of circuits to execute (which in this case would be the trap circuit with the target circuit hidden among them), by Robert, the single-qubit gates are redacted. Therefore, as the target circuits and the trap circuits only differ in their single-qubit gates (as assumed in Sec.~\ref{sec:TrapsAndTargets}) Bob cannot distinguish trap circuit executions from target circuit executions at any stage in Protocol~\ref{StandardAccAlg}. The redaction additionally allows for the unbreakable encryption of the output of each trap or target, as in Algorithm~\ref{ModifiedTrapGenAlg}, so Bob cannot locate the target via the outputs either.

With the probability of error in any target circuit, $\mathcal{C}_{\text{targ}}$, executed within the same single use of Protocol~\ref{adversarialExecutionProtocol} bounded -- as in Eqn.~\ref{eqn:boundOnTrapErrorFirstAppearanceinTheorem} -- a bound on the ideal-actual variation distance of the execution of that target circuit, $\nu \big[ \Tilde{\mathcal{C}}_{\text{targ}} \big]$, can be obtained. By following Ref.~\cite[Appendix Sec. 1]{Ferracin_2021}, I derive that the ideal-actual variation distance of the execution of that target circuit is upper bounded by the probability of error occurring in that execution. Therefore,
\begin{align}
        \nu \big[ \Tilde{\mathcal{C}}_{\text{targ}} \big] 
        \leq
        \dfrac{ 1 + \beta }{k (1 - \beta)} \big( \bar{v} + \theta \big).
    \end{align}
\end{proof}

\section{Discussion}
\label{discussionSec}
In this paper, I have upgraded the accreditation protocol of Ref.~\cite{Ferracin_2021} to consider all noise/error to be adversarial. This has necessitated the development of a new model of adversarial error (i.e. the adapted cryptographic setting presented in Sec.~\ref{sec:ErrorModel}), where Bob is limited according to experimentally derived rules. These rules (that single-qubit gates are hidden from Bob and that the probability of error in similar circuits executed in quick succession have similar probabilities of error), encapsulated in Protocol~\ref{adversarialExecutionProtocol}, allow my upgraded accreditation protocol to function almost identically when experiencing adversarial error as the one in Ref.~\cite{Ferracin_2019} does when the error is IID. Therefore, the desirable qualities the protocol from Ref.~\cite{Ferracin_2019}, such as not requiring lengthening the circuit to accredit, or adding excessively many extra single-qubit gates, or any extra two-qubit gates, or ancilla qubits, or trusting any aspect of a computation, or extra connectivity between qubits are preserved; leading to no diminution in its suitability for near-term usage.

The principle question left open by this paper is to what extent the limits on Bob can be relaxed. Relaxing the limitations would have to entail changes to Protocol~\ref{adversarialExecutionProtocol}. I believe the redaction of single-qubit gates cannot be eliminated without substantial changes to the trap, as otherwise Bob could distinguish target and trap circuits, although it is not clear trap and target circuits that are indistinguishable to Bob can be constructed without requiring another limitation on Bob.

A slight improvement can be obtained using Ref.~\cite[Theorem 1]{Ferracin_2021}, which allows the error in single-qubit gates to be weakly gate dependent: this could be added to my cryptographic model by allowing Bob to define error in the -- still redacted -- single-qubit gates that depends on the single-qubit gates\footnote{Single-qubit gates would still have to be redacted to prevent Bob identifying the target but he could, perhaps, specify some dependence without seeing the single-qubit gates e.g. ``if that gate is an $\hat{X}$ gate...".} but the differing CPTP maps -- for each single-qubit gate -- must differ, in the diamond norm, by at most some small known value.
This is a further step towards reality and the restriction on the differences between the gate-dependent error reflects the very small errors in single-qubit gates~\cite[Fig.5]{9355264}. 
  
\section{Acknowledgements and Declarations}
I would like to thank Animesh Datta and Theodoros Kapourniotis for useful conversations.
I would also like to thank an anonymous reviewer for suggesting a shorter proof of Lemma 3.
This work was supported, in part, 
by an EPSRC IAA grant (G.PXAD.0702.EXP), the UKRI ExCALIBUR project QEVEC (EP/W00772X/2), and the Quantum Advantage Pathfinder (EP/X026167/1).

The author declares no conflicts of interest
\newpage
\bibliography{References2}
\end{document}